\documentclass[conference]{IEEEtran}
\usepackage{hyperref} 
\usepackage{flushend}
\usepackage{graphicx}
\usepackage{amssymb}
\newtheorem{remark}{Remark}

\newtheorem{lemma}{Lemma}

\newtheorem{coro}{Corollary}

\newenvironment{proof}{{\noindent\it Proof:}}{\hfill $\blacksquare$\par}
\newcommand{\dd}{\mathrm{d}}

\usepackage{cite}

\ifCLASSINFOpdf
\else
\fi
\usepackage[cmex10]{amsmath}
\begin{document}
%
\title{{Matched-filter Precoded Rate Splitting Multiple Access: A Simple and Energy-efficient Design}}
\author{\IEEEauthorblockN{Hui Zhao, and Dirk Slock}
\IEEEauthorblockA{Communication Systems Department,
EURECOM, Sophia Antipolis, France\\
Email: hui.zhao@eurecom.fr; dirk.slock@eurecom.fr}\vspace{-0.8cm}
}


%


\maketitle

\begin{abstract}
We introduce an energy-efficient downlink rate splitting multiple access (RSMA) scheme, employing a simple matched filter (MF) for precoding. We consider a transmitter equipped with multiple antennas, serving several single-antenna users at the same frequency-time resource, each with distinct message requests. Within the conventional 1-layer RSMA framework, requested messages undergo splitting into common and private streams, which are then precoded separately before transmission. In contrast, we propose a novel strategy where only an MF is employed to precode both the common and private streams in RSMA, promising significantly improved energy efficiency and reduced complexity.  We demonstrate that this MF-precoded RSMA achieves the same delivery performance as conventional RSMA, where the common stream is beamformed using maximal ratio transmission (MRT) and the private streams are precoded by MF. Taking into account imperfect channel state information at the transmitter, we proceed to analyze the delivery performance of the MF-precoded RSMA. We derive the ergodic rates for decoding the common and private streams at a target user respectively in the massive MIMO regime. Finally, numerical simulations validate the accuracy of our analytical models, as well as demonstrate the advantages over conventional RSMA.
\end{abstract}


%
\IEEEpeerreviewmaketitle

\section{Introduction}
Conventional multi-access (MA) schemes like space division MA (SDMA) heavily rely on timely and highly accurate channel state information at the transmitter (CSIT), which poses significant challenges in real wireless communications, particularly for high-speed moving users and hardware impairments. To tackle these challenges, rate-splitting multiple access (RSMA) has emerged as a potent and resilient non-orthogonal transmission framework, efficiently leveraging time, frequency, power, spatial, and code-domain resources \cite{Bruno_Clerckx,Clerckx_Open,Clerckx_Mag}.
RSMA refers to a broad class of multi-user (MU) schemes relying on the rate-splitting (RS) principle. RS involves partitioning each message into multiple segments for inter-user interference management, allowing for flexible decoding by receivers \cite{Mishra_CL}.  The flexibility of RSMA has been shown to suit the inherent dynamics of wireless networks resulting from network density, topology and load, interference levels, and so forth.

Consider the downlink 1-layer RSMA (cf. \cite{Salem,Mishra_CL,Joudeh}) illustrated in Fig.~\ref{1layer_RSMA_fig}, where a multi-antenna transmitter concurrently serves $K$ single-antenna users, each with distinct message requests. Each message is first divided into two components: the common segment and the private segment. Subsequently, the common segments derived from all requested messages are aggregated to form a collective common message beneficial for all the $K$ users. Meanwhile, the private segments are individually encoded into $K$ private streams, each dedicated to a specific user, before undergoing processing by a precoder. The beamformed common stream and the precoded private streams are finally superposed for transmission. On the receiver end, the common stream is prioritized for decoding, allowing the receiver to eliminate interference stemming from the common stream through 1-layer successive interference cancellation (SIC) before decoding the intended private stream. Notably, by adjusting the power distribution between the common and private streams, we gain the flexibility to manage interference levels at receivers effectively.

In contrast to conventional multi-user multiple-input multiple-output (MU-MIMO) systems, RSMA requires the transmitter to design a beamformer for the common stream and a precoder for the private streams separately. This undoubtedly leads to significantly higher energy consumption and increased hardware costs. Furthermore, the optimal beamformer for the common stream, which must address the max-min fairness issue, necessitates solving a non-convex and NP-hard optimization problem in each channel realization \cite{Yuan_Wang}.  Although non-linear dirty paper coding (DPC) can achieve the MU-MIMO channel region, simple linear precoders such as zero-forcing (ZF), regularized ZF (RZF), and matched-filter (MF) are commonly employed in current MU-MIMO systems due to their practicality \cite{Rusek,LuLu,Wagner}. As a result, there has been an increasing interest in developing simple and effective beamformers and linear precoders individually tailored for the common and private streams in RSMA (cf. \cite{Salem,Lei_TCOM,Dai}). Despite these efforts, the integrated design of the beamformer and precoder continues to present more complexity than conventional linear precoding schemes.

It is important to highlight that both ZF and RZF require the inversion of the channel matrix \cite{Wagner}, which demands significant computational resources, particularly in scenarios with massive connectivity. In contrast, MF simply involves taking the conjugate of the channel matrix, resulting in substantially lower complexity \cite{Ngo_TCOM,Feng,Zhao_TWC}. Moreover, real wireless channels often exhibit rank deficiency, further amplifying the computational burden associated with matrix inversion in ZF precoding. Consequently, MF is commonly preferred in some practical implementations. However, it is well-known that MF is interference-limited, causing transmission rates to plateau in the high signal-to-noise ratio (SNR) regime. The work in \cite{Minjue} demonstrates that RSMA effectively controls the interference level associated with MF precoding by judiciously allocating power between the common and private streams, consequently resulting in increased overall throughput.  Furthermore, insights from \cite{Shang} suggest that utilizing an MF matrix for precoding the common stream can outperform a vector beamformer in terms of transmission rate. However, it is essential to acknowledge that this approach involves processing signals with several times more dimensions than those handled by the vector beamformer.

This paper presents a novel precoding scheme for the downlink 1-layer RSMA, wherein an MF matrix is utilized to precode both the common and private streams. This approach eliminates the need for designing a dedicated beamformer for the common stream and allows for better control of interference levels in MF for the private streams. Then, we analyze the corresponding delivery performance in the widely considered massive MIMO regime (cf. \cite{Onur_Dizdar,Dai}). The \emph{technical contributions} are outlined below.

\begin{itemize}
    \item We will demonstrate that MF-precoded RSMA achieves equivalent delivery performance compared to the conventional RSMA, where maximum
ratio transmission (MRT) and MF are employed to process the common and private streams separately.
    \item In the presence of imperfect CSIT, we will rigorously establish that the transmission rate for the common stream converges in distribution to a specific random variable under an average power constraint.
    \item Once more, under imperfect CSIT, we will rigorously prove that the rate for sending a private stream converges almost surely to a constant.\footnote{To the best of the authors’ knowledge, \emph{this is the first time to present a rigorous proof to reveal the convergence results of MRT and MF} under an average transmit power constraint in the massive MIMO regime. We refer to \cite{Feng,Zhao_TWC,Ngo_TCOM} and \cite{Onur_Dizdar,QiZhang,Chae} for some prior efforts.}
    \item Utilizing the aforementioned convergence results, we will derive the ergodic rates individually for transmitting the common and private streams to a target user.\footnote{\emph{Notations.} We use $\mathbb{C}$ to denote the complex number set, and use ${\bf I}_L$ to denote the $L \times L$ identity matrix. For a matrix ${\bf A}$, we use ${\bf A}^H$, ${\bf A}^T$ and ${\bf A}^*$ to denote the conjugate transpose, the non-conjugate transpose and the conjugate part of ${\bf A}$ respectively. $|\cdot|$ denotes the magnitude of a complex number, while $||\cdot||$ denotes the norm-2 of a vector. $\mathbb{E}\{\cdot\}$ and ${\rm Tr}\{\cdot\}$ represent the expectation operator and the trace operator respectively. ${\rm Diag}\{{\bf a}\}$ denotes the diagonal matrix where the elements in $\bf a$ are the diagonal elements. ${\bf 1}_L \in \mathbb{C}^{L \times 1}$ stands for the vector with all elements equaling 1, while ${\bf 0}_L \times \mathbb{C}^{L \times 1}$ is the vector with all zero elements. $\mathcal{CN}$ denotes the complex Gaussian distribution. We use $\stackrel{d.}{\longrightarrow} $, $\stackrel{a.s.}{\longrightarrow} $ and $\stackrel{\text{L}^1}{\longrightarrow} $ to denote the convergence in distribution, the almost sure convergence and the convergence in the mean of order 1 respectively.}
\end{itemize}

\begin{figure*}
    \centering
    \includegraphics[width= 5 in]{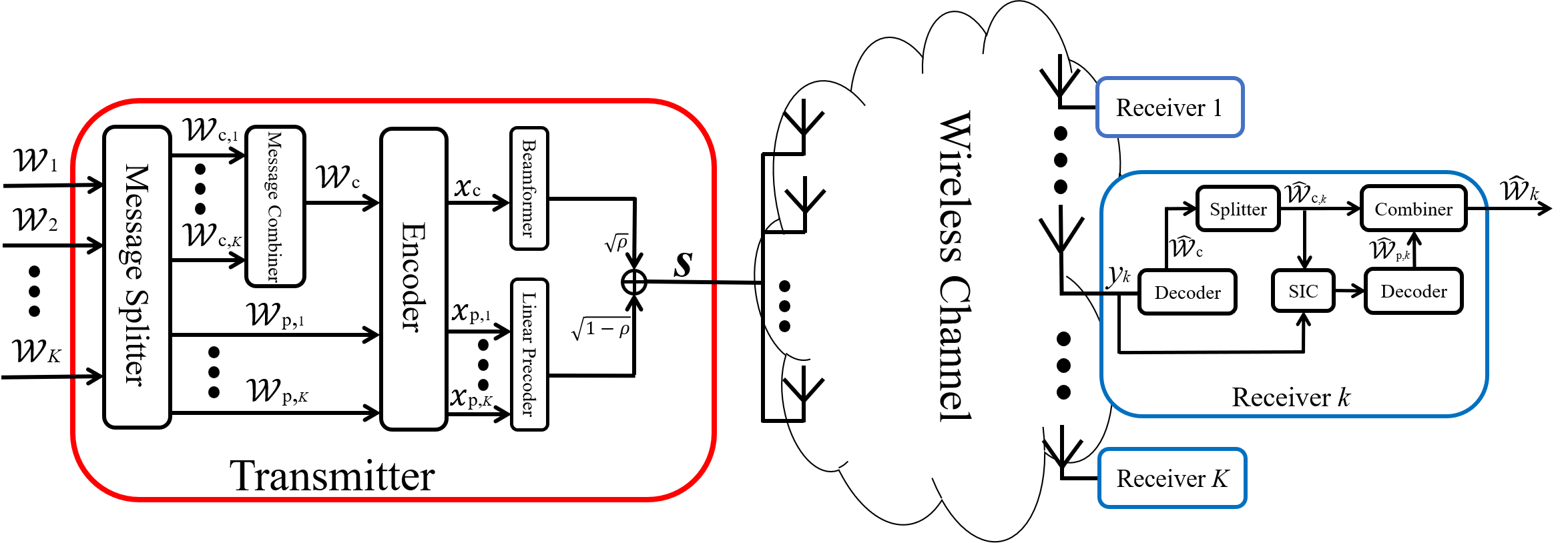}
    \caption{Conventional 1-layer RSMA with linear precoder for downlink 
 transmission}\vspace{-0.3cm}\label{1layer_RSMA_fig}
\end{figure*}

\section{System Model}
In the downlink transmission, a base station (BS) has $L$ transmit antennas to serve $K$ single-antenna users at a time, with each user requesting different messages.

\subsection{Conventional 1-layer RSMA With Linear Precoder}
In this subsection, we will elaborate on the conventional 1-layer downlink RSMA under linear precoding (cf. \cite{Mishra_CL}). Let $\mathcal{W}_k$ denote the message intended by the $k$-th user for $k = 1,2, \cdots,K$.  As depicted in Fig.~\ref{1layer_RSMA_fig}, under the 1-layer RSMA framework, each message $\mathcal{W}_k$ will be first split into two parts, i.e., the common part $\mathcal{W}_{c,k}$ and the private part $\mathcal{W}_{p,k}$. Then all the common parts generated from $K$ messages $\{\mathcal{W}_{c,1}, \mathcal{W}_{c,2},\cdots, \mathcal{W}_{c,K} \}$ will be combined together to form a common message, which is useful for all users. The $K$ private parts will be \emph{independently} mapped into $K$ private streams, each private stream associating with a dedicated user. Let $x_k$ represent the private stream for the $k$-th user, and let $x_c$ represent the common stream for all the users.  A beamformer ${\bf w}_c \in \mathbb{C}^{L \times 1}$ will be carefully designed for broadcasting the common stream to the served $K$ users, while the private stream vector ${\bf x}_p \triangleq [x_{p,1}, x_{p,2}, \cdots, x_{p,K}]^T$ will be precoded by a matrix ${\bf W}_p = [{\bf w}_{p,1}, {\bf w}_{p,2},\cdots, {\bf w}_{p,K}] \in \mathbb{C}^{L \times K}$. Let ${\bf W} = [ {\bf w}_c \  {\bf W}_p ] \in \mathbb{C}^{L \times (K+1)}$ and ${\bf x} = [x_c, {\bf x}_p^T]^T$. The transmit signal ${\bf s} \in \mathbb{C}^{L \times 1}$ at the BS can be mathematically written as
\begin{align}\label{signal_RSMA}
    {\bf s}  = \sqrt{\alpha}{\bf W} {\bf P} {\bf x} =  \sqrt{\alpha \rho} {\bf w}_c x_c + \sqrt{\alpha (1-\rho)} {\bf W}_p {\bf x}_p,
\end{align}
where $\rho \in [0,1]$ is the power-splitting factor between the common and the private streams\footnote{We note that the total power of $(1 - \rho)P_t$ allocated to the $K$ private streams can be numerically optimized to enhance the delivery performance. In this paper, we adhere to the equal-power allocation as adopted in \cite{Salem,Shang}.}, and ${\bf P} \triangleq \text{Diag}\{ \sqrt{\rho}, \sqrt{1-\rho}, \cdots, \sqrt{1-\rho} \} \in \mathbb{C}^{(K+1)\times (K+1)} $. In \eqref{signal_RSMA},
$\alpha \triangleq \frac{P_t}{{\rm Tr} \{ \mathbb{E}\{{\bf W}^H {\bf W}\} {\bf P}^2  \}}$, denoted as the power normalization factor, serves to confine the transmit power within an average constraint of $P_t$. It is easy to check that $
    \mathbb{E}\{||{\bf s}||^2\} = \alpha \mathbb{E}\big\{ {\bf x}^H {\bf P} {\bf W}^H {\bf W} {\bf P} {\bf x} \big\} 
     = \alpha {\rm Tr}\big\{ \mathbb{E}\{{\bf W}^H {\bf W}\} {\bf P}^2  \big\} = P_t.$

The received signal at the $k$-th user then takes the form
\begin{align}
    y_k = \sqrt{\alpha \rho} {\bf h}_k^T {\bf w}_c x_c + \sqrt{\alpha \bar \rho} {\bf h}_k^T {\bf W}_p {\bf x}_p + n_k,
\end{align}
where $\bar \rho = 1-\rho $, and $n_k \sim \mathcal{CN}(0,\sigma_k^2)$ denotes the Additive White Gaussian noise (AWGN).  Under the usual Gaussian signalling assumption, the signal-to-interference-plus-noise ratio (SINR) for decoding the common stream at the $k$-th user is of the form\footnote{ To maintain low complexity and alleviate the need for intricate numerical optimizations, MRT (${\bf w}_c^{\rm mrt} = \sum_{i=1}^K {\bf h}_k^*$) is often employed to efficiently perform beamforming for the common stream in RSMA (cf. \cite{Salem,Minjue}).}
\begin{align}
    \text{SINR}_{c,k} = \frac{\alpha \rho |{\bf h}_k^T {\bf w}_c|^2}{\sigma_k^2 + \alpha \bar \rho \sum_{i=1}^K |{\bf h}_k^T {\bf w}_{p,i}|^2 }
\end{align}
After decoding the common stream successfully, the $k$-th user can remove it from the received signal via 1-layer SIC before decoding the intended private stream $x_{p,k}$. Therefore, the SINR for decoding $x_{p,k}$ at the $k$-th user  takes the form
\begin{align}
    \text{SINR}_{p,k} = \frac{\alpha \bar \rho |{\bf h}_k^T {\bf w}_{p,k}|^2 }{\sigma_k^2 + \alpha \bar \rho \sum_{i=1, i\neq k}^K |{\bf h}_k^T {\bf w}_{p,i}|^2 }
\end{align}
The sum rate of this RSMA system is of the form \cite{Salem,Mishra_CL}
\begin{align}
    R_{\rm sum} = \log_2 \big( 1 + \text{SINR}_{c}\big) + \sum_{k=1}^K \log_2 \big( 1 + \text{SINR}_{p,k} \big),
\end{align}
where $\text{SINR}_{c} \triangleq \min\{ \text{SINR}_{c,1}, \cdots, \text{SINR}_{c,K}\}$ guarantees the successful decoding of $x_c$ at all served users.

\subsection{MF-precoded RSMA Design}
\begin{figure}[!t]
        \centering
        \includegraphics[width= 2.8 in]{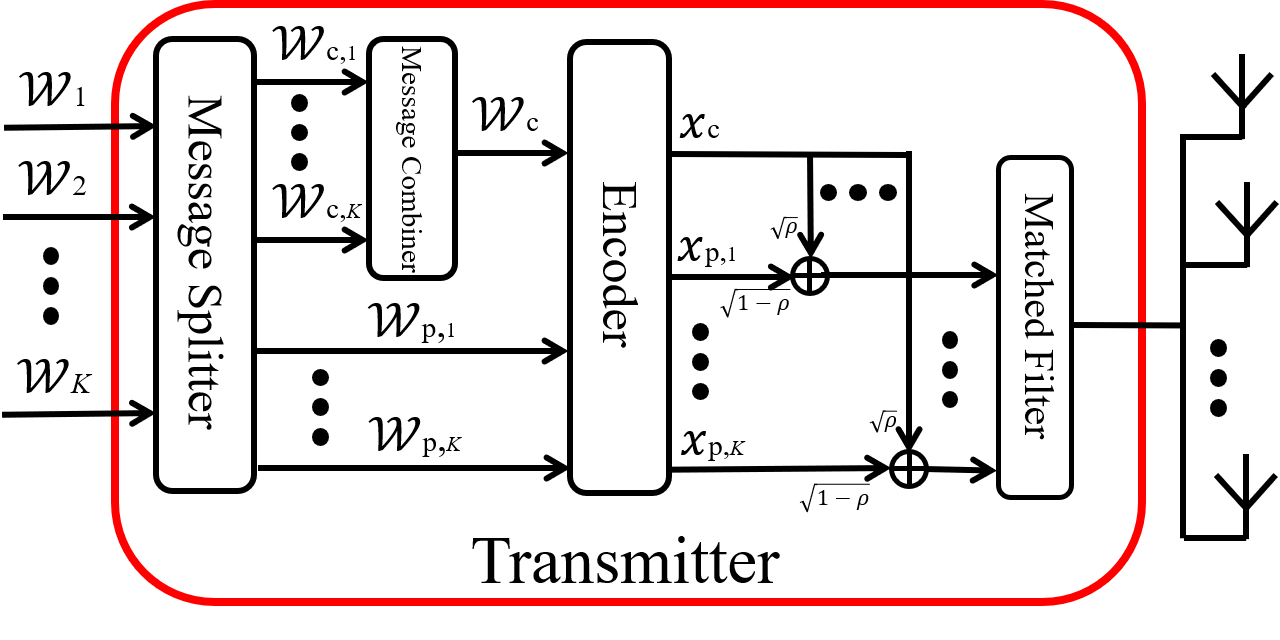}
        \caption{The transmitter structure of proposed MF-precoded RSMA}\label{MF_precode_RSMA_fig}\vspace{-0.3cm}
\end{figure}	
The proposed MF-precoded RSMA scheme is depicted in Fig.~\ref{MF_precode_RSMA_fig}. In contrast to designing a beamformer dedicated to the common stream, we superpose the common and private streams together before entering an MF precoder\footnote{We note that  MF can be substituted with any high-performance precoder based on numerical optimization.  However, we confine ourselves to employing a simple MF for low complexity in this paper. We also note that MF has been shown to provide very high spectral efficiency which often comes close to the performance of DPC when the BS has a large number of antennas \cite{Rusek}.}. For the real channel matrix ${\bf H} \triangleq [ {\bf h}_1, {\bf h}_2,\cdots, {\bf h}_K ]^T \in \mathbb{C}^{K \times L}$, the transmit signal is mathematically written as
\begin{align}\label{trans_sig_MF}
    {\bf s}_{\rm mf} &= \sqrt{\alpha_{\rm mf}} {\bf H}^H {\bf x}_{\rm mf} = \sqrt{\alpha_{\rm mf}} \sum_{i=1}^K {\bf h}_i^* \big( \sqrt{\rho} x_c + \sqrt{\bar \rho} x_{p,i} \big),
\end{align}
where ${\bf x}_{\rm mf} \triangleq \sqrt{\rho} {\bf 1}_K  x_c + \sqrt{\bar \rho}{\bf x}_p$ is the signal before MF precoding, and again $\rho \in [0,1]$ and $\alpha_{\rm mf}$ denote the power-splitting factor and the power normalization factor respectively. The received signal at the $k$-th user is then of the form
\begin{align}
    &y_k^{\rm mf} 
    =\sqrt{\alpha_{\rm mf} } {\bf h}_k^T \sum_{i=1}^K  \hat {\bf h}_i^* \big( \sqrt{\rho} x_c  + \sqrt{\bar \rho} {x}_{p,i} \big) + n_k \notag\\
    &\hspace{0.5cm}= \sqrt{\alpha_{\rm mf} \rho} \bigg( {\bf h}_k^T \sum_{i=1}^K   {\bf h}_i^* \bigg) x_c + \sqrt{\alpha_{\rm mf} \bar \rho} \sum_{i=1}^K {\bf h}_k^T {\bf h}_i^* x_{p,i} + n_k, \notag
\end{align}
The SINR for decoding $x_c$ at the $k$-th user takes the form
\begin{align}\label{SINR_ck_def_eq}
    \text{SINR}_{c,k}^{\rm mf}  = \frac{\alpha_{\rm mf} \rho \big| {\bf h}_k^T \sum_{i=1}^K   {\bf h}_i^* \big|^2}{\sigma_k^2 + \alpha_{\rm mf} \bar \rho \sum_{i=1}^K  | {\bf h}_k^T {\bf h}_i^* |^2 }.
\end{align}
So $R_{c,k}^{\rm mf} = \log_2 ( 1 + \text{SINR}_{c,k}^{\rm mf} )$ is the maximum transmission rate that the $k$-th user can decode $x_c$ successfully. 
After decoding the common stream, the $k$-th user can remove the common stream from the received signal $y_k$ via SIC before decoding the intended private stream. Therefore, the SINR for decoding $x_{p,k}$ at the $k$-th user is of the form
\begin{align}\label{SINR_pk_def_eq}
    \text{SINR}_{p,k}^{\rm mf} = \frac{\alpha_{\rm mf} \bar \rho |{\bf h}_k^T {\bf h}_k^*|^2 }{\sigma_k^2 + \alpha_{\rm mf}\bar \rho \sum_{i=1,i\neq k}^K |{\bf h}_k^T {\bf h}_i^*|^2 }.
\end{align}
The maximum transmission rate for sending the $k$-th private stream is $R_{p,k}^{\rm mf} = \log_2 (1+ \text{SINR}_{p,k}^{\rm mf})$.

It is easy to see from \eqref{SINR_ck_def_eq} and \eqref{SINR_pk_def_eq} that the SINRs separately for decoding the common and intended private streams at a target user have the same forms as those in the MRT-MF precoded RSMA. So we have the following statement.
\begin{remark}\label{MF_rem}
    Under the average power constraint of $P_t$, the MF-preceded RSMA achieves the same delivery performance as the conventional 1-layer RSMA where MRT beamforms the common stream and MF precodes the $K$ private streams.
\end{remark}

\section{Performance Analysis on MF-Precoded RSMA}
In this section, we will analyze the delivery performance of the MF-precoded RSMA in the massive MIMO regime. 

\subsection{Analytical Model}
In this subsection, we will present some basic assumptions for performance analysis.
The channel vector ${\bf h}_k$ from the BS to the $k$-th user follows a multivariate complex Gaussian distribution, i.e., ${\bf h}_k \sim \mathcal{CN}({\bf 0}_L, \beta_k {\bf I}_L)$, where $\beta_k$ accounts for the large-scale pathloss and/or shadowing \cite{Ngo_TCOM}.
We also assume that the downlink channel training is perfect for the receivers, while the channel feedback to the BS appears with some errors, i.e., imperfect CSIT.  For the maximum likelihood (ML) estimator in TDD training, the ML estimate of ${\bf h}_k$ is $\hat {\bf h}_k = {\bf h}_k + \tilde {\bf h}_k$ where $\tilde {\bf h}_k \sim \mathcal{CN}({\bf 0}_L, \tilde \beta_k {\bf I}_L)$ is the estimation error, which is independent of ${\bf h}_k$ (cf. \cite{Arti}). Obviously, the variance of each element in $\hat {\bf h}_k$ is $\hat \beta_k = \beta_k + \tilde \beta_k$.

In this analytical model, we consider the \emph{delay-tolerant transmission} such that sending the common stream and the $k$-th private stream at ergodic rates guarantees successful decoding at the $k$-th user (cf. \cite{Salem,Minjue,Joudeh}). Accordingly, the ergodic sum rate (ESR) takes the form
\begin{align}
    \bar R_{\rm sum}^{\rm mf} = \min_{k=1,\cdots,K} \big\{ \bar R^{\rm mf}_{c,k} \big\} + \sum_{k=1}^K \bar R^{\rm mf}_{p,k},
\end{align}
where $\bar R^{\rm mf}_{c,k} \triangleq \mathbb{E}\{R^{\rm mf}_{c,k}\}$ and $\bar R^{\rm mf}_{p,k} \triangleq \mathbb{E}\{R^{\rm mf}_{p,k}\}$ with the expectation over channel states.

In the imperfect CSIT case,  the transmit signal at the BS under the MF-precoded RSMA scheme becomes
\begin{align}
    {\bf s}_{\rm mf} = \sqrt{\alpha_{\rm mf}} \hat {\bf H}^H {\bf x}_{\rm mf},
\end{align}
where $\hat {\bf H} \triangleq [ \hat{\bf h}_1, \hat{\bf h}_2, \cdots, \hat{\bf h}_K ]^T$ is the estimated channel at the BS. For the power normalization factor, after some mathematical manipulations, we can easily derive that
\begin{align}
    \alpha_{\rm mf} = \frac{P_t}{L \sum_{k=1}^K \hat \beta_k},
\end{align}
which is the same as the conventional MF precoding under an average power constraint of $P_t$ (cf. \cite{Feng}).

The received signal at the $k$-th user becomes
\begin{align}
    y_k^{\rm mf} 
    = \sqrt{\alpha_{\rm mf} \rho} \bigg( {\bf h}_k^T \sum_{i=1}^K  \hat {\bf h}_i^* \bigg) x_c + \sqrt{\alpha_{\rm mf} \bar \rho} \sum_{i=1}^K {\bf h}_k^T \hat {\bf h}_i^* x_i + n_k,  \notag
\end{align}
which can be further written as \eqref{received_sig_mf}, shown at the top of the next page.
\begin{figure*}
    \begin{align}\label{received_sig_mf}
        y_k^{\rm mf} = \sqrt{\alpha_{\rm mf} \rho} \bigg( {\bf h}_k^T \sum_{i=1}^K  {\bf h}_i^* \bigg) x_c + \sqrt{\alpha_{\rm mf} \rho} \bigg( {\bf h}_k^T \sum_{i=1}^K  \tilde {\bf h}_i^* \bigg) x_c
         +  \sqrt{\alpha_{\rm mf} \bar \rho} \sum_{i=1}^K {\bf h}_k^T  {\bf h}_i^* x_i +  \sqrt{\alpha_{\rm mf} \bar \rho} \sum_{i=1}^K {\bf h}_k^T \tilde {\bf h}_i^* x_i  + n_k
    \end{align}
  \rule{18cm}{0.01cm}  
\end{figure*}
The SINR for decoding $x_c$ is of the form\footnote{We note that each receiver has perfect CSI for decoding, so the impact of imperfect CSIT should appear in the numerator of the SINR expression (cf. \cite{Salem}).  We also note that each receiver can perform perfect SIC in delay-tolerant transmission with imperfect CSIT, provided that the common and private streams are transmitted at their ergodic rates (cf. \cite{Salem,Joudeh}).}
\begin{align}\label{SINR_ck_def}
    \text{SINR}_{c,k}^{\rm mf} = \frac{\alpha_{\rm mf} \rho \big| {\bf h}_k^T \sum_{i=1}^K  {\bf h}_i^* \big|^2 + \alpha_{\rm mf} \rho \sum_{i=1}^K \tilde \beta_i ||{\bf h}_k||^2  }{\sigma_k^2 + \alpha_{\rm mf} \bar \rho \sum_{i=1}^K \big( \big| {\bf h}_k^T {\bf h}_i^* \big|^2 + \tilde \beta_i ||{\bf h}_k||^2  \big)}.
\end{align}
After removing the common stream from the received signal, the signal for decoding $x_{p,k}$ becomes
\begin{align}
    y_{p,k}^{\rm mf} = \sqrt{\alpha_{\rm mf} \bar \rho} \sum_{i=1}^K \big( {\bf h}_k^T  {\bf h}_i^* + {\bf h}_k^T \tilde {\bf h}_i^*  \big)x_i   + n_k
\end{align}
The SINR for decoding $x_{p,k}$ takes the form
\begin{align}\label{SINR_pk_def}
    \text{SINR}_{p,k}^{\rm mf} = \frac{\alpha_{\rm mf} \bar \rho || {\bf h}_k ||^4 + \alpha_{\rm mf} \bar \rho \tilde \beta_k ||{\bf h}_k||^2 }{\sigma_k^2 + \alpha_{\rm mf} \bar \rho \sum_{i=1, i\neq k}^K \big( |{\bf h}_k^T {\bf h}_i^*|^2 + \tilde \beta_i ||{\bf h}_k||^2 \big) }.
\end{align}

\subsection{Convergence Results in Massive MIMO}
In this subsection, we will derive some convergence results \emph{based on rigorous proof} in the massive MIMO regime where both $L$ and $K$ go to infinity with the ratio $\theta = L/K$ fixed.
Before presenting the main results, let us first define some parameters as follows. 
\begin{align}
    \beta_{\rm ave} \triangleq \lim_{K \to \infty}\frac{1}{K} \sum_{i=1}^K \beta_i, \ \ \ \hat \beta_{\rm ave} \triangleq \lim_{K \to \infty} \frac{1}{K} \sum_{i=1}^K \hat \beta_i,
\end{align}
which are both finite.
We define their ratio as $\delta \triangleq \beta_{\rm ave}/ \hat \beta_{\rm ave} \in [0,1]$.
Let $X$ denote a random variable that follows a Noncentral chi-squared distribution with the degrees-of-freedom of $2$ and the non-centrality parameter of $\frac{2 \theta \beta_k}{\beta_{\rm ave}}$, denoted by $X  \sim \chi^2 \big(2,  \frac{2 \theta \beta_k}{\beta_{\rm ave}} \big)$. The probability density function (PDF) of $X$ is of the form
  \begin{align}\label{PDF_chi}
        f_{X} (x) = \frac{1}{2} \exp\Big( - \frac{x }{2} - \frac{\theta \beta_k}{\beta_{\rm ave}} \Big) \ \mathcal{I}_0 \Bigg( \sqrt{\frac{2 \theta \beta_k}{\beta_{\rm ave}} x} \Bigg),
    \end{align}
where $\mathcal{I}_\cdot\{\cdot\}$ denotes the modified Bessel functions of the first kind.  
Now, we can present the following result for $R_{c,k}^{\rm mf}$.
\begin{lemma}\label{SINR_ck_lem}
    As $L,K \to \infty$ with a fixed ratio $\theta$ under imperfect CSIT, the transmission rate for the common stream to the $k$-th user has the convergence result:\footnote{Unlike the convergence results observed in MF, ZF and RZF (cf. \cite{Rusek}) where their rates converge to constants, we uncover an interesting fact: the transmission rate of MRT (cf. Remark~\ref{MF_rem}) converges to a specific random distribution when $L$ and $K$ both approach infinity while with a fixed ratio.}
    \begin{align}\label{R_ck_ergodic_eq}
        R_{c,k}^{\rm mf} \stackrel{d.}{\longrightarrow}  \log_2 \bigg( 1 + \frac{   \beta_k \rho P_t \big( \frac{\delta}{2}X + 1-\delta \big) } { \sigma_k^2 + \bar \rho  \beta_k P_t \big( 1 + \theta \beta_k/\hat \beta_{\rm ave} \big)  }  \bigg).
    \end{align}
\end{lemma}

\begin{proof}
    The proof is relegated to Appendix~\ref{SINR_ck_lem_proof}.
\end{proof}

\begin{coro}\label{R_ck_coro}
    As $L,K \to \infty$ with a fixed ratio $\theta$, the ergodic rate for decoding the common stream at the $k$-th user takes the form in \eqref{Rck_eq}, shown at the top of the next page.
    \begin{figure*}
     \begin{align}\label{Rck_eq}
            \lim_{L,K\to \infty} \bar R_{c,k}^{\rm mf} =  \frac{1}{2} \int_0^\infty \log_2 \Bigg( 1 +    \frac{   \beta_k \rho P_t \big( \frac{\delta}{2} x + 1-\delta \big) } { \sigma_k^2 + \bar \rho  \beta_k P_t \big( 1 + \theta \beta_k/\hat \beta_{\rm ave} \big)  }\Bigg)   \exp\bigg(- \frac{x }{2}  - \frac{\theta \beta_k}{\beta_{\rm ave}} \bigg)  \ \mathcal{I}_0 \Bigg( \sqrt{\frac{2 \theta \beta_k}{\beta_{\rm ave}} x} \Bigg) \dd x
        \end{align}
         \rule{18cm}{0.01cm}
        \end{figure*}
\end{coro}
\begin{proof}
For $\text{SINR}_{c,k}^{\rm mf}$ in \eqref{SINR_ck_rewrite}, we can have that
\begin{align}
   \text{SINR}_{c,k}^{\rm mf} &\le \frac{\rho}{\bar \rho} \frac{ \big|\sum_{i=1}^K  {\bf u}_k^T{\bf h}_i^* \big|^2 + \sum_{i=1}^K \tilde \beta_i  }{    \sum_{i=1}^K  \big| {\bf u}_k^T {\bf h}_i^* \big|^2 +   \sum_{i=1}^K \tilde \beta_i} \notag\\
   & \le \frac{\rho}{\bar \rho} \frac{ \sum_{i=1}^K  \big| {\bf u}_k^T {\bf h}_i^* \big|^2 + \sum_{i=1}^K \tilde \beta_i  }{    \sum_{i=1}^K  \big| {\bf u}_k^T {\bf h}_i^* \big|^2 +   \sum_{i=1}^K \tilde \beta_i} \!= \frac{\rho}{\bar \rho} < \infty,
\end{align}
which implies that $R_{c,k}^{\rm mf}$ is always uniformly integrable. In this case, convergence in distribution implies the convergence in expectation. That is, we have that
  \begin{align}
        &\lim_{L,K\to \infty} \bar R_{c,k}^{\rm mf} = \lim_{L,K\to \infty} \mathbb{E}\{ R_{c,k}^{\rm mf} \}\notag\\
        &\hspace{0.5cm}= \mathbb{E}\bigg\{ \log_2 \bigg( 1 +  \frac{   \beta_k \rho P_t \big( \frac{\delta}{2}X + 1-\delta \big) } { \sigma_k^2 + \bar \rho \beta_k P_t \big( 1 + \theta \beta_k/\hat \beta_{\rm ave} \big)  }\bigg) \bigg\}. 
    \end{align}
    Substituting the PDF of $X$ (cf. \eqref{PDF_chi}) into the above, we arrive at  \eqref{Rck_eq}, which concludes the proof. 
\end{proof}

Next, we will present the convergence result for $R_{p,k}^{\rm mf}$.
\begin{lemma}\label{R_pk_lem}
    In the massive MIMO regime, the rate for sending the $k$-th private stream converges to a constant almost surely,
    \begin{align}
       R_{p,k}^{\rm mf}  \stackrel{a.s.}{\longrightarrow}  \log_2 \bigg( 1 + \frac{\beta_k}{\hat \beta_{\rm ave}} \frac{\theta \bar \rho P_t  \beta_k}{\sigma_k^2  + \bar \rho P_t  \beta_k } \bigg).
    \end{align}
\end{lemma}

\begin{proof}
Dividing both the numerator and denominator of $\text{SINR}_{p,k}^{\rm mf}$ in \eqref{SINR_pk_def} by $||{\bf h}_k||^2 K^{-1}$, we obtain that
   \begin{align}\label{SINR_pk_rewrite}
    \text{SINR}_{p,k}^{\rm mf} = \frac{K\alpha_{\rm mf} \bar \rho || {\bf h}_k ||^2 + K\alpha_{\rm mf} \bar \rho  \tilde \beta_k  }{ \frac{\sigma_k^2 K}{|| {\bf h}_k ||^2} + K \alpha_{\rm mf} \bar \rho  \sum_{i=1, i\neq k}^K \big( \big| {\bf u}_k^T {\bf h}_i^* \big|^2 + \tilde \beta_i  \big) }
\end{align}
For $L,K \to \infty$, using the Strong Law of Large Numbers (SLLN), we can easily have that
\begin{align}
K \alpha_{\rm mf} || {\bf h}_k ||^2 
\stackrel{a.s.}{\longrightarrow} \frac{P_t \beta_k}{\hat \beta_{\rm ave}}, \ 
\frac{K}{||{\bf h}_k||^2}  \stackrel{a.s.}{\longrightarrow} \frac{1}{\theta \beta_k}, \ 
 K \alpha_{\rm mf} 
\stackrel{a.s.}{\longrightarrow} 0. \notag
\end{align}
Considering the independence statement in \eqref{indepen_proof_eq}, we have that
\begin{align}
    &K \alpha_{\rm mf} \sum_{i=1, i\neq k}^K \Big( \big| {\bf u}_k^T {\bf h}_i^* \big|^2 + \tilde \beta_i  \Big) \notag\\
    & \hspace{0.5cm}= \frac{P_t}{L \hat \beta_{\rm ave}} \frac{K-1}{K-1} \bigg( \sum_{i=1, i\neq k}^K \Big( \big| {\bf u}_k^T {\bf h}_i^* \big|^2 -  \beta_i \Big)   + \sum_{i=1,i\neq k}^K \hat \beta_i  \bigg)  \notag\\
    & \hspace{0.5cm} \stackrel{a.s.}{\longrightarrow}  \frac{P_t }{\theta \hat \beta_{\rm ave}} \Big( 0 + \hat \beta_{\rm ave} \Big) = \frac{P_t}{\theta}.
\end{align}
Combining the convergence results listed above easily yields Lemma~\ref{R_pk_lem}, which concludes the proof.
\end{proof}

\begin{coro}\label{R_pk_coro}
    In the massive MIMO regime, the ergodic rate for decoding the $k$-th private stream at the dedicated user is of the form\footnote{We note that various approximations (e.g., \cite{Rusek,Feng,QiZhang,Chae,Zhao_TWC,Ngo_TCOM}) have been developed for the ergodic rate of MF precoding in massive MIMO. Interestingly, these approximations can finally arrive at  \eqref{R_pk_ergodic_eq}. However, this is the first time to rigorously prove this well-known convergence result under the average power constraint. }
    \begin{align}\label{R_pk_ergodic_eq}
        \lim_{L,K \to \infty} \bar R_{p,k}^{\rm mf} = \log_2 \bigg( 1 + \frac{\beta_k}{\hat \beta_{\rm ave}} \frac{\theta \bar \rho P_t  \beta_k}{\sigma_k^2  + \bar \rho P_t  \beta_k } \bigg).
    \end{align}
\end{coro}

\begin{proof}
    We first have that
    \begin{align}
        R_{p,k}^{\rm mf} = \frac{\ln ( 1 + \text{SINR}_{p,k}^{\rm mf} )}{\ln 2}  \le \frac{\text{SINR}_{p,k}^{\rm mf}}{\ln 2} 
    \end{align}
    Based on \eqref{SINR_pk_rewrite}, we can easily prove that $\mathbb{E}\{(\text{SINR}_{p,k}^{\rm mf}\big)^2 \} < \infty$, so $R_{p,k}^{\rm mf}$ is uniformly integrable. 
    According to the Vitali Convergence Theorem and Lemma~\ref{R_pk_lem}, we have that
    \begin{align}
        R_{p,k}^{\rm mf}  \stackrel{\text{L}^1}{\longrightarrow}  \log_2 \bigg( 1 + \frac{\beta_k}{\hat \beta_{\rm ave}} \frac{\theta \bar \rho P_t  \beta_k}{\sigma_k^2  + \bar \rho P_t  \beta_k } \bigg),
    \end{align}
    which directly yields \eqref{R_pk_ergodic_eq}, and which concludes the proof.
\end{proof}

\section{Numerical Results}
In this section, we conduct numerical experiments to showcase the precision of the derived formulas and to provide comparisons with other RSMA designs.

In Figs. \ref{ESR_L_fig_ref}--\ref{ESR_N_fig_ref}, we examine the symmetric fading scenario, where each user shares identical channel statistics. Moreover, $\beta_k = \sigma_k^2 = 1$ is assumed for any $k=1,2,\cdots,K$. For the channel estimation model, we consider that $\tilde \beta_k = (P_t N)^{-1}$, where $N$ is the number of training symbols, and which enables an enhancement in the quality of channel estimation as the SNR rises. As shown in Fig. \ref{ESR_L_fig_ref}, ESR exhibits a linear growth as both $L$ and $K$ increase with a fixed ratio $\theta$. As $P_t$ increases, the ESR also rises, attributable to the heightened transmit power and enhanced channel estimation.
We also observe a close alignment between the analytical results obtained from Corollaries \ref{R_ck_coro} and \ref{R_pk_coro} and the simulated outcomes, particularly in the regime of large $L$. However, this correspondence is not as strong for small $L$ and high transmit power. The discrepancy primarily arises from Corollary \ref{R_pk_coro} regarding the ergodic rate under conventional MF precoding, which proves to be an inadequate approximation in the medium to high SNR range for finite numbers of $L$ and $K$ (cf. \cite{Feng,Zhao_MF}).\footnote{Unfortunately, the magnitude of this discrepancy becomes too significant to disregard in the simulation settings in Figs. \ref{ESR_N_fig_ref}--\ref{ESR_Com_L16_ref}. We note that a more precise approximation can be obtained by employing the methods outlined in our latest work \cite{Zhao_MF}, but we have reserved this for follow-up works. For that, we utilize the simulated result for  $\bar R_{p,k}^{\rm mf}$ while employing Corollary \ref{R_ck_coro} for $\bar R_{c,k}^{\rm mf}$ to derive the analytical results presented in Figs. \ref{ESR_N_fig_ref}--\ref{ESR_Com_L16_ref}.}

\begin{figure}[!t]
             \vspace{-0.2cm}
             \centering
             \includegraphics[width= 3.3 in]{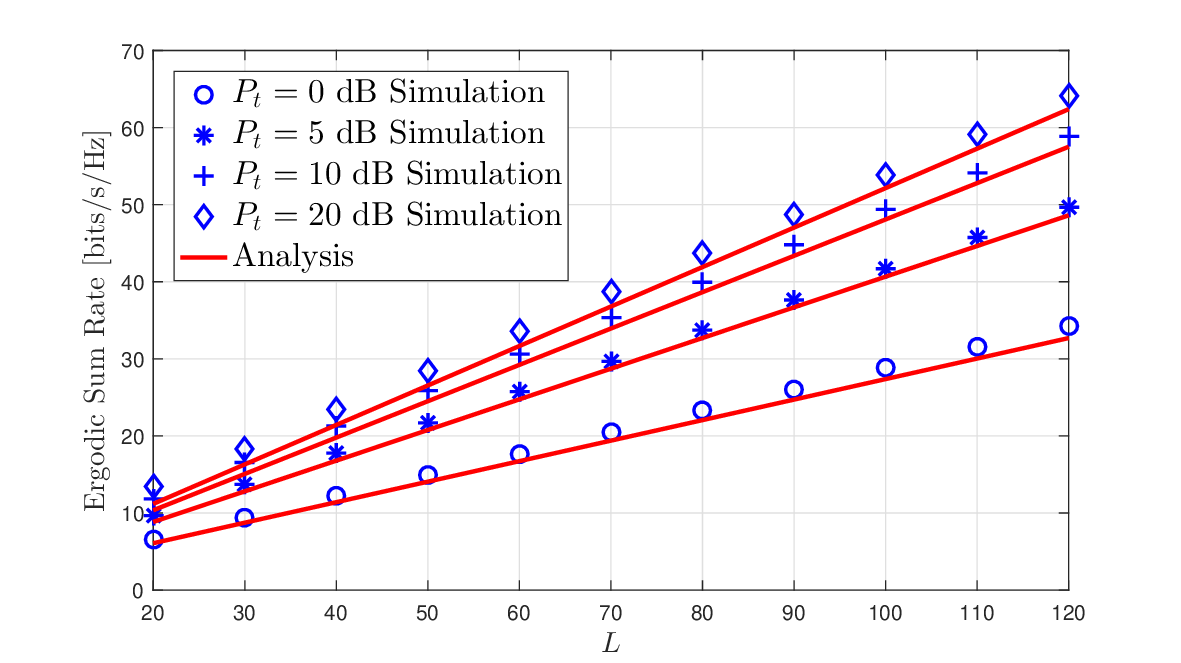}
             \vspace{-0.3cm}
        \caption{Ergodic sum rate versus $\rho$ for $\theta=5$, $\rho=0.5$, and $N=10$.}\vspace{-0.3cm}\label{ESR_L_fig_ref}
\end{figure}

In Fig.~\ref{ESR_N_fig_ref}, the ESR is plotted against $\rho$ for various values of $N$. It is evident that employing more training symbols results in a higher ESR in the low and medium $\rho$ regime. However, as $\rho$ approaches 1 (indicating purely MU multicasting), the ESRs across different channel estimation qualities converge. This suggests that the decoding of the common stream is not significantly affected by the estimation quality. Compared to conventional MF precoding, RSMA enables us to control the interference level between the common and private streams, thereby achieving a higher rate by adjusting the value of $\rho$.

\begin{figure}[!t]
             \vspace{-0.2cm}
             \centering
             \includegraphics[width= 3.3 in]{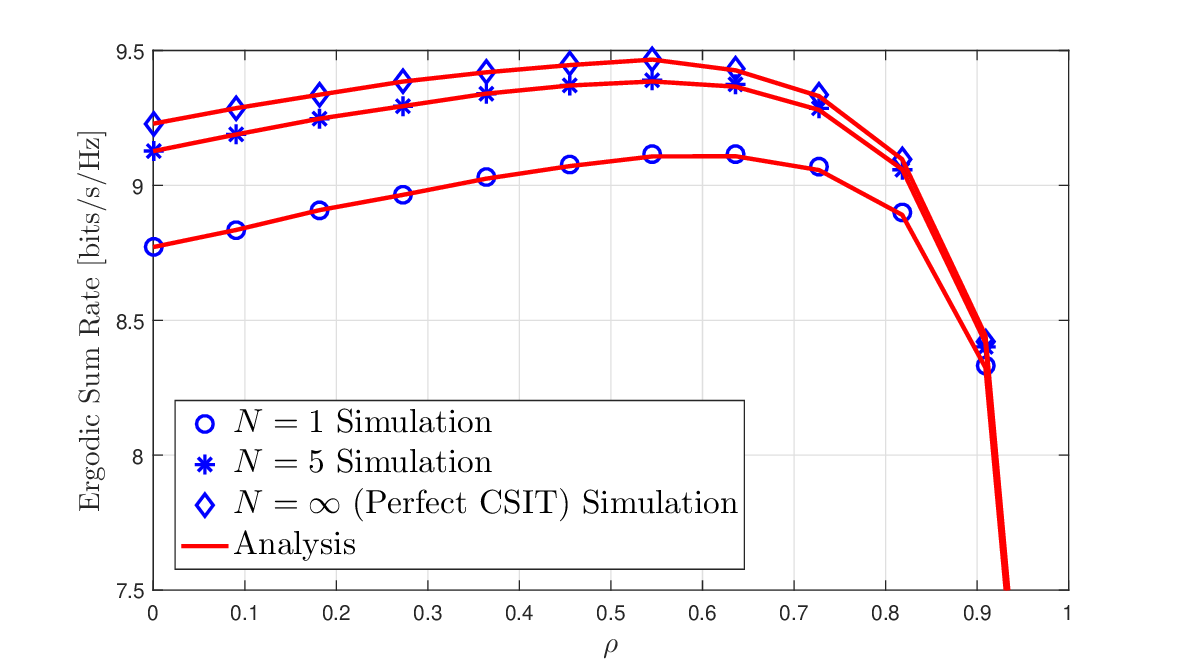}
             \vspace{-0.3cm}
        \caption{Ergodic sum rate versus $\rho$ for $L=12$, $K=4$, and  $P_t=10$ dB.}\vspace{-0.3cm}\label{ESR_N_fig_ref}
\end{figure}

In Figs. \ref{ESR_Com_L8_ref}--\ref{ESR_Com_L16_ref}, we explore a realistic wireless propagation scenario within a Macro-cell environment \cite{Emil_PathlossModel}. Specifically, we consider an AWGN spectral density of $-174$ dBm/Hz, with a spectrum bandwidth of 20 MHz allocated for each user. We generate 1000 realizations of user locations, assuming a uniform distribution of users across a Macro-cell with an inner radius of 35 meters and an outer radius of 500 meters. 
Assuming a carrier frequency of 2 GHz, we have that  $\beta_k = \mathcal{L}_0 r_k^{-\eta}$, where $r_k$ is the distance from the BS to the $k$-th user, $\eta=3.76$ denotes the pathloss exponent, and $\mathcal{L}_0 = 10^{-3.53}$ regulates the channel attenuation at 35 meters \cite{Emil_PathlossModel}. 
We assume $P_t = 40$ dBm, which is a typical transmit power at a Macro-cell BS \cite{GSMA_guide}. Additionally, we employ ZF and RZF precoders separately for the private streams, while the common stream is always beamformed by MRT in the conventional 1-layer RSMA (cf. \cite{Salem}). These represent two benchmark schemes for the proposed MF-precoded RSMA. To ease the simulation, we consider perfect CSIT where $\tilde \beta_k =0$.

In Figs. \ref{ESR_Com_L8_ref}--\ref{ESR_Com_L16_ref}, we observe that MF-precoded RSMA outperforms MRT-ZF preceded RSMA when $L = K = 8$, whereas the comparison is reversed when the number of transmit antennas increases from 8 to 16. As anticipated, MRT-RZF consistently delivers the best performance. Additionally, more numerical results (omitted due to space constraints) fully demonstrate that MF often outperforms ZF for $\theta = L/K \le 1$. It is worth noting that in dense Macro-cell environments, the number of served users can reach several thousand \cite{GSMA_guide}, making scenarios where $\theta = L/K \le 1$ quite common.

\begin{figure}[!t]
             \vspace{-0.2cm}
             \centering
             \includegraphics[width= 3.3 in]{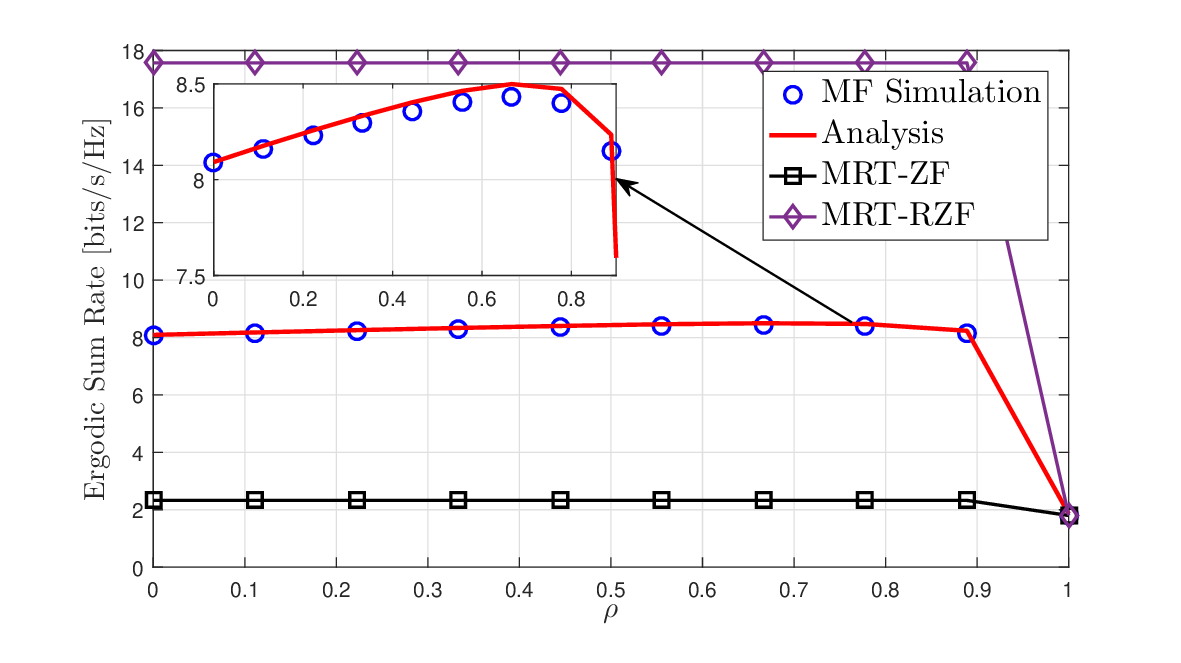}
             \vspace{-0.5cm}
        \caption{Ergodic sum rate versus $\rho$ for $L=K=8$.}\vspace{-0.2cm}\label{ESR_Com_L8_ref}
\end{figure}

\begin{figure}[!t]
             \vspace{-0.2cm}
             \centering
             \includegraphics[width= 3.3 in]{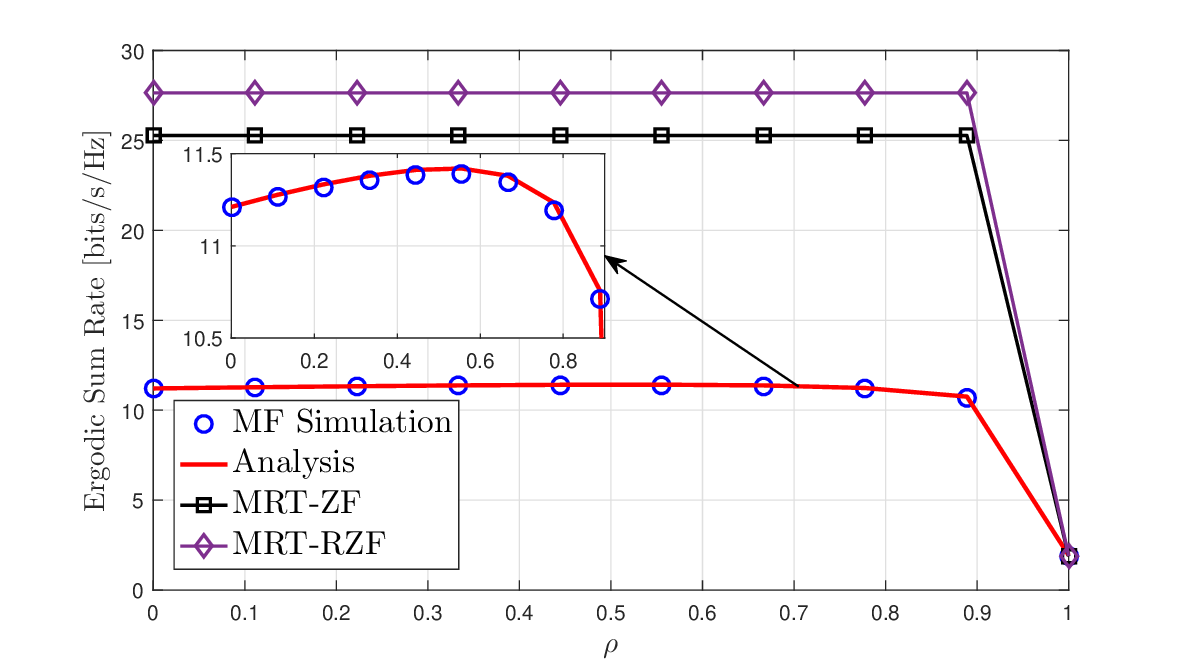}
             \vspace{-0.3cm}
        \caption{Ergodic sum rate versus $\rho$ for $L=16$ and $K=8$.}\vspace{-0.2cm}\label{ESR_Com_L16_ref}
\end{figure}

\section{Conclusion}
In this paper, we have developed an MF-precoded RSAM scheme, which offers significantly improved energy efficiency and reduced complexity. We have demonstrated that the MF-precoded RSMA exhibits the same delivery performance as the conventional MRT-MF precoded RSMA. In the massive MIMO regime and in imperfect CSIT, we have rigorously proved that the transmission rate for the common stream converges in distribution to a specific random variable, while the transmission rate for each private stream converges to a constant almost surely. Based on the convergence results, we derived the ergodic rates separately for decoding the common and private streams.  Finally, our numerical results have not only validated the accuracy of our analytical models but have also showcased the advantages of our proposed scheme over alternative approaches.



\appendices
\renewcommand{\thesectiondis}[2]{\Alph{section}:}

\section{Proof of Lemma~\ref{SINR_ck_lem}}\label{SINR_ck_lem_proof}
Dividing both the numerator and denominator of $\text{SINR}_{c,k}^{\rm mf}$ in \eqref{SINR_ck_def} by $||{\bf h}_k||^2 K^{-1}$, we obtain that
    \begin{align}\label{SINR_ck_rewrite}
       \text{SINR}_{c,k}^{\rm mf} = \frac{K\alpha_{\rm mf} \rho \big|{\bf u}_k^T\sum_{i=1}^K  {\bf h}_i^* \big|^2 + K\alpha_{\rm mf} \rho \sum_{i=1}^K \tilde \beta_i  }{\frac{K \sigma_k^2}{||{\bf h}_k||^2} + \bar \rho K \alpha_{\rm mf}  \sum_{i=1}^K \big( \big| {\bf u}_k^T {\bf h}_i^* \big|^2 + \tilde \beta_i   \big)},
    \end{align}
where ${\bf u}_k \triangleq {\bf h}_k /||{\bf h}_k||$ with unit-norm. For any $i, j \in \{1,2,\cdots,K\}$ and $i \neq j \neq k$,  we have that ${\bf u}_k^T {\bf h}_i \sim \mathcal{CN}(0, \beta_i)$ and ${\bf u}_k^T {\bf h}_j \sim \mathcal{CN}(0, \beta_j)$. We also have that
\begin{align}\label{indepen_proof_eq}
    &\mathbb{E}\big\{ {\bf u}_k^T {\bf h}_i {\bf h}_j^H   {\bf u}_k^* \big\}
    = {\rm Tr}\big\{   \mathbb{E}\big\{  {\bf h}_i   \big\} \mathbb{E}\big\{   {\bf h}_j^H  \big\}  \mathbb{E}\big\{ {\bf u}_k^* {\bf u}_k^T \big\}\big\} = 0,
\end{align}
which implies that ${\bf u}_k^T {\bf h}_i $ and ${\bf u}_k^T {\bf h}_j$ are uncorrelated. As ${\bf u}_k^T {\bf h}_i $ and ${\bf u}_k^T {\bf h}_j$ are both Gaussian distributed, they are independent of each other. Therefore, we have that
$
        \sum_{i=1,i\neq k}^K {\bf u}_k^T  {\bf h}_i^* \sim \mathcal{CN}\big(0, b \big),
$
where $b \triangleq \sum_{i=1, i\neq k}^K \beta_{i}$.  
Define the random variable $X'$ as
$
        X' \triangleq {\bf u}_k^T \sum_{i=1}^K  {\bf h}_i^* = ||{\bf h}_k|| + \sum_{i=1,i\neq k}^K {\bf u}_k^T  {\bf h}_i^*.
$
    It is easy to derive that $||{\bf h}_k|| $ follows a Nakagami distribution with the shape parameter $L$ and the spread parameter $\beta_k L $, denoted by $||{\bf h}_k|| \sim \text{Nakagami}(L,\beta_k L)$.   So $X'$ is the sum of a Nakagami distributed variable and a complex Gaussian distributed variable, which has the same definition as the channel gain over Shadowed-Rician fading channels (cf. \cite{Slim_Rician,Zhao_ICC}). 
    For that, we have the moment generating function (MGF) of $ X'$ as (cf. \cite[Eq. (7)]{Slim_Rician})
    \begin{align}
        \mathcal{M}_{X}(t) 
        & = \frac{1}{1-bt} \bigg( 1 -   \frac{\beta_k t}{1-bt}  \bigg)^{-L}, \ \text{  for } t < \frac{1}{b}. 
    \end{align}   
    Let $X = \frac{2}{b} X'$. The MGF of $X$ is of the form
    \begin{align}
        \mathcal{M}_{X}(t) = \mathcal{M}_{X'}\Big( \frac{2}{b} t  \Big) 
        = \frac{1}{1-2t} \bigg( 1 -   \frac{2\beta_k  t}{1-2t} \frac{1}{b} \bigg)^{-L},
    \end{align}
    where $t < \frac{1}{2}$.
    As $b = (\sum_{i=1}^K \beta_i) - \beta_k = K \beta_{\rm ave} - \beta_k = \frac{L}{\theta} \beta_{\rm ave} - \beta_k$. For $L,K \to \infty$ with a fixed ratio $\theta$, we can derive the limit MGF of $X$ as
    \begin{align}
        \mathcal{M}_{X}(t) 
        & = \frac{1}{1-2t} \lim_{L,K \to \infty} \bigg( 1 -   \frac{2\theta \beta_k  t/\beta_{\rm ave}}{1-2t} \frac{1}{ L  - \theta \beta_k/\beta_{\rm ave}} \bigg)^{-L} \notag\\
        & = \frac{1}{1-2t} \exp\bigg( \frac{2\theta \beta_k}{\beta_{\rm ave}} \frac{  t}{1-2t} \bigg), \text{ for } t < \frac{1}{2}
    \end{align}
    which equals the MGF of $\chi^2 \big(2,  \frac{2\theta \beta_k}{\beta_{\rm ave}} \big)$. Therefore,
    $
        X \stackrel{d.}{\longrightarrow} \chi^2 \big(2,  \frac{2\theta \beta_k}{\beta_{\rm ave}} \big).
   $
   For other terms in the numerator of \eqref{SINR_ck_rewrite}, we have that
  \begin{align}
      & \lim_{L,K \to \infty}  \frac{K \alpha_{\rm mf} b}{2} = \lim_{L,K \to \infty} \frac{K}{2}  \frac{P_t}{L K \hat \beta_{\rm ave} }  \Big( \frac{L}{\theta} \beta_{\rm ave} - \beta_k \Big) \notag\\
      &\hspace{1cm}= \lim_{L,K \to \infty} \frac{P_t}{2} \Big( \frac{\beta_{\rm ave}}{\theta \hat \beta_{\rm ave}} - \frac{\beta_k}{L \hat \beta_{\rm ave}} \Big) = \frac{ \beta_{\rm ave} P_t}{2 \theta \hat \beta_{\rm ave}}, \\
      & K \alpha_{\rm mf}  \sum_{i=1}^K \tilde \beta_k = \frac{P_t K}{L  \hat \beta_{\rm ave}} \frac{1}{K} \sum_{i=1}^K \tilde \beta_k= \frac{P_t \big( \hat \beta_{\rm ave} - \beta_{\rm ave} \big)}{\theta \hat \beta_{\rm ave}}. 
  \end{align}
  
 As $L,K \to \infty$ with a fixed ratio $\theta$, using the Strong Law of Large Numbers (SLLN) yields that
    \begin{align}
        &\frac{K}{||{\bf h}_k||^2}  \stackrel{a.s.}{\longrightarrow} \frac{1}{\theta \beta_k}, \ 
        K \alpha_{\rm mf} ||{\bf h}_k||^2   \stackrel{a.s.}{\longrightarrow} \frac{P_t \beta_k}{\hat \beta_{\rm ave}} \notag\\
        &\frac{1}{K-1} \sum_{i=1,i\neq k}^K \Big( \big| {\bf u}_k^T {\bf h}_i^* \big|^2 - \beta_i \Big)  \stackrel{a.s.}{\longrightarrow} 0,
    \end{align}
    which helps us derive the convergence result for the denominator of \eqref{SINR_ck_rewrite}, given by
    \begin{align}
         \frac{K \sigma_k^2}{||{\bf h}_k||^2} &+ \bar \rho K \alpha_{\rm mf}  \sum_{i=1}^K \big( \big| {\bf u}_k^T {\bf h}_i^* \big|^2 + \tilde \beta_i   \big) \notag\\
        &\stackrel{a.s.}{\longrightarrow} \frac{ \sigma_k^2}{\theta \beta_k} + \frac{\bar \rho P_t}{\theta} \Big( \frac{\theta \beta_k}{ \hat \beta_{\rm ave}} + 1 \Big).
    \end{align}

    All the terms (except $X$) in \eqref{SINR_ck_rewrite} converge to constants almost surely, so $\text{SINR}_{c,k}^{\rm mf}$ has the following convergence:   
    \begin{align}
        \text{SINR}_{c,k}^{\rm mf} 
        &\stackrel{d.}{\longrightarrow} \frac{\frac{ \beta_{\rm ave} \beta_k \rho P_t}{2  \hat \beta_{\rm ave}} X + \frac{\rho  P_t \beta_k}{\hat \beta_{\rm ave}}  \big( \hat \beta_{\rm ave} - \beta_{\rm ave} \big)}{\sigma_k^2 + \bar \rho  \beta_k P_t \big( \frac{\theta \beta_k}{ \hat \beta_{\rm ave}} + 1  \big) }. 
    \end{align}
    Considering that $R_{c,k}^{\rm mf} = \log_2 (1+\text{SINR}_{c,k}^{\rm mf})$ and then using the Continuous Mapping Theorem, we can derive the convergence result for $R_{c,k}^{\rm mf}$ in \eqref{R_ck_ergodic_eq}.

	\bibliographystyle{IEEEtran}				
	\bibliography{IEEEabrv,aBiblio}			

\begin{thebibliography}{10}
\providecommand{\url}[1]{#1}
\csname url@samestyle\endcsname
\providecommand{\newblock}{\relax}
\providecommand{\bibinfo}[2]{#2}
\providecommand{\BIBentrySTDinterwordspacing}{\spaceskip=0pt\relax}
\providecommand{\BIBentryALTinterwordstretchfactor}{4}
\providecommand{\BIBentryALTinterwordspacing}{\spaceskip=\fontdimen2\font plus
\BIBentryALTinterwordstretchfactor\fontdimen3\font minus
  \fontdimen4\font\relax}
\providecommand{\BIBforeignlanguage}[2]{{%
\expandafter\ifx\csname l@#1\endcsname\relax
\typeout{** WARNING: IEEEtran.bst: No hyphenation pattern has been}%
\typeout{** loaded for the language `#1'. Using the pattern for}%
\typeout{** the default language instead.}%
\else
\language=\csname l@#1\endcsname
\fi
#2}}
\providecommand{\BIBdecl}{\relax}
\BIBdecl

\bibitem{Bruno_Clerckx}
{B. Clerckx, \emph{et al.}}, ``A primer on rate-splitting multiple access:
  Tutorial, myths, and frequently asked questions,'' \emph{IEEE J. Sel. Areas
  Commun.}, vol.~41, no.~5, pp. 1265--1308, May 2023.

\bibitem{Clerckx_Open}
{B. Clerckx \emph{et al.}}, ``Is {NOMA} efficient in multi-antenna networks?
  {A} critical look at next generation multiple access techniques,'' \emph{IEEE
  Open J. Commun. Soc.}, vol.~2, p. 1310–1343, Jun. 2021.

\bibitem{Clerckx_Mag}
{B. Clerckx, H. Joudeh, C. Hao, M. Dai, and B. Rassouli}, ``Rate splitting for
  {MIMO} wireless networks: A promising {PHY}-layer strategy for {LTE}
  evolution,'' \emph{IEEE Commun. Mag.}, vol.~54, no.~5, pp. 98--105, May 2016.

\bibitem{Mishra_CL}
{A. Mishra, Y. Mao, O. Dizdar and B. Clerckx}, ``Rate-splitting multiple access
  for {6G}—part {I}: Principles, applications and future works,'' \emph{IEEE
  Commun. Lett.}, vol.~26, no.~10, pp. 2232--2236, Oct. 2022.

\bibitem{Salem}
{A. Salem, C. Masouros, and B. Clerckx}, ``Secure rate splitting multiple
  access: How much of the split signal to reveal?'' \emph{IEEE Trans. Wireless
  Commun.}, vol.~22, no.~6, pp. 4173--4187, Jun. 2023.

\bibitem{Joudeh}
{H. Joudeh and B. Clerckx}, ``Sum-rate maximization for linearly precoded
  downlink multiuser {MISO} systems with partial {CSIT}: {A} rate-splitting
  approach,'' \emph{IEEE Trans. Commun.}, vol.~64, no.~11, pp. 4847--4861, Nov.
  2016.

\bibitem{Yuan_Wang}
Y.~{Wang}, V.~W.~S. {Wong}, and J.~{Wang}, ``Flexible rate-splitting multiple
  access with finite blocklength,'' \emph{IEEE J. Sel. Areas Commun.}, vol.~41,
  no.~5, pp. 1398--1412, May 2023.

\bibitem{Rusek}
F.~{Rusek}, D.~{Persson}, B.~K. {Lau}, E.~G. {Larsson}, T.~L. {Marzetta},
  O.~{Edfors}, and F.~{Tufvesson}, ``Scaling up {MIMO}: {O}pportunities and
  challenges with very large arrays,'' \emph{IEEE Signal Process. Mag.},
  vol.~30, no.~1, pp. 40--60, Jan. 2013.

\bibitem{LuLu}
L.~{Lu}, G.~Y. {Li}, A.~L. {Swindlehurst}, A.~{Ashikhmin}, and R.~{Zhang}, ``An
  overview of massive {MIMO}: {B}enefits and challenges,'' \emph{IEEE J. Sel.
  Topics Signal Process.}, vol.~8, no.~5, pp. 742--758, Oct. 2014.

\bibitem{Wagner}
S.~{Wagner}, R.~{Couillet}, M.~{Debbah}, and D.~T.~M. {Slock}, ``Large system
  analysis of linear precoding in correlated {MISO} broadcast channels under
  limited feedback,'' \emph{IEEE Trans. Inf. Theory}, vol.~58, no.~7, pp.
  4509--4537, Jul. 2012.

\bibitem{Lei_TCOM}
H.~{Lei}, S.~{Zhou}, K.-H. {Park}, I.~S. {Ansari}, H.~{Tang}, and M.-S.
  {Alouini}, ``Outage analysis of millimeter wave {RSMA} systems,'' \emph{IEEE
  Trans. Commun.}, vol.~71, no.~3, pp. 1504--1520, Mar. 2023.

\bibitem{Dai}
{M. Dai, B. Clerckx, D. Gesbert and G. Caire}, ``A rate splitting strategy for
  massive {MIMO} with imperfect {CSIT},'' \emph{IEEE Trans. Wireless Commun.},
  vol.~15, no.~7, pp. 4611--4624, Jul. 2016.

\bibitem{Ngo_TCOM}
{H. Q. Ngo, E. G. Larsson, and T. L. Marzetta}, ``Energy and spectral
  efficiency of very large multiuser {MIMO} systems,'' \emph{IEEE Trans.
  Commun.}, vol.~61, no.~4, pp. 1436--1449, Apr. 2013.

\bibitem{Feng}
C.~{Feng}, Y.~{Jing}, and S.~{Jin}, ``Interference and outage probability
  analysis for massive {MIMO} downlink with {MF} precoding,'' \emph{IEEE Signal
  Process. Lett.}, vol.~23, no.~3, pp. 366--370, Mar. 2016.

\bibitem{Zhao_TWC}
{H. Zhao \emph{et al.}}, ``Vector coded caching multiplicatively increases the
  throughput of realistic downlink systems,'' \emph{IEEE Trans. Wireless
  Commun.}, vol.~22, no.~4, pp. 2683--2698, Apr. 2023.

\bibitem{Minjue}
M.~{He}, H.~{Zhao}, X.~{Miao}, S.~{Wang}, and G.~{Pan}, ``Secure rate-splitting
  multiple access transmissions in {LMS} systems,'' \emph{IEEE Commun. Lett.},
  vol.~28, no.~1, pp. 19--23, Jan. 2024.

\bibitem{Shang}
W.~{Shang}, P.~{Ren}, Z.~{Xiang}, and L.~{Yang}, ``An exact analysis on
  matrix-based precoding for the common message in downlink rate splitting
  multiple access,'' \emph{IEEE Commun. Lett.}, to be published. DOI:
  10.1109/LCOMM.2024.3359290.

\bibitem{Onur_Dizdar}
O.~{Dizdar}, Y.~{Mao}, and B.~{Clerckx}, ``Rate-splitting multiple access to
  mitigate the curse of mobility in (massive) {MIMO} networks,'' \emph{IEEE
  Trans. Common.}, vol.~69, no.~10, pp. 6765--6780, Oct. 2021.

\bibitem{QiZhang}
{Q. Zhang, S. Jin, K.-K. Wong, H. Zhu, and M. Matthaiou}, ``Power scaling of
  uplink massive {MIMO} systems with arbitrary-rank channel means,'' \emph{IEEE
  J. Sel. Topics Signal Process}, vol.~8, no.~5, pp. 966--981, Oct. 2014.

\bibitem{Chae}
{Y.-G. Lim, C.-B. Chae, and G. Caire}, ``Performance analysis of massive {MIMO}
  for cell-boundary users,'' \emph{IEEE Trans. Wireless Commun.}, vol.~14,
  no.~12, p. 6827–6842, Dec. 2015.

\bibitem{Arti}
M.~K. Arti, ``Imperfect {CSI} based maximal ratio combining in
  shadowed-{R}ician fading land mobile satellite channels,'' in \emph{Proc.
  Twenty 1st Nat. Conf. Commun. (NCC)}, Feb 2015.

\bibitem{Zhao_MF}
\BIBentryALTinterwordspacing
H.~{Zhao}, D.~{Slock}, and P.~{Elia}, ``Exact {SINR} analysis of matched-filter
  precoder,'' submitted to \emph{Proc. IEEE Int. Symp. Inf. Theory (ISIT)}.
  [Online]. Available: \url{https://doi.org/10.48550/arXiv.2401.16958}
\BIBentrySTDinterwordspacing

\bibitem{Emil_PathlossModel}
E.~{Björnson}, M.~{Kountouris}, and M.~{Debbah}, ``Massive {MIMO} and small
  cells: {I}mproving energy efficiency by optimal soft-cell coordination,'' in
  \emph{Proc. Int. Conf. Telecommun. (ICT)}, May 2013.

\bibitem{GSMA_guide}
``{5G Implementation Guidelines},'' GSMA, Tech. Rep. version 2.0, 2019.

\bibitem{Slim_Rician}
A.~{Abdi}, W.~C. {Lau}, M.-S. {Alouini}, and M.~{Kaveh}, ``A new simple model
  for land mobile satellite channels: First- and second-order statistics,''
  \emph{IEEE Trans. Wireless Commun.}, vol.~2, no.~3, pp. 519--528, May 2003.

\bibitem{Zhao_ICC}
{H. Zhao, A. Bazco-Nogueras, and P. Elia}, ``Coded caching in land mobile
  satellite systems,'' in \emph{Proc. IEEE Int. Conf. Commun. (ICC)}, May 2022.

\end{thebibliography}

\end{document}